\documentclass[a4paper,article,10pt]{article}
\usepackage{amssymb}
\usepackage{amstext}
\usepackage{amsmath}
\usepackage{amscd}
\usepackage{amsthm}
\usepackage{amsfonts}
\usepackage{enumerate}
\usepackage{latexsym}
\usepackage{mathrsfs}
\usepackage{mathtools}

\makeatletter
 
 \@addtoreset{equation}{section}
\makeatother

\newtheorem{theorem}{Theorem}[section]
\newtheorem{lemma}[theorem]{Lemma}
\newtheorem{proposition}[theorem]{Proposition}
\newtheorem{remark}[theorem]{Remark}

\DeclareMathOperator{\ind}{ind}

\DeclareMathOperator{\wn}{Wind}

\DeclareMathOperator{\Trace}{Tr}
\DeclareMathOperator{\trace}{tr}

\DeclareMathOperator{\Ln}{Ln}
\DeclareMathOperator{\CAP}{CAP}
\DeclareMathOperator{\APS}{APS}
\DeclareMathOperator{\APL}{APL}
\DeclareMathOperator{\Bohr}{B}
\DeclareMathOperator{\discrete}{discrete}
\DeclareMathOperator{\ap}{a.p.}
\DeclareMathOperator{\Index}{Index}

\def\C{\mathbb C}
\def\N{\mathbb N}
\def\S{\mathbb S}
\def\R{\mathbb R}
\def\Z{\mathbb Z}
\def\A{\mathcal A}
\def\B{\mathcal B}
\def\d{\mathrm d}
\def\D{\mathcal D}
\def\e{\mathrm e}
\def\E{\mathcal E}
\def\F{\mathscr F}
\def\H{\mathcal H}
\def\I{\mathcal I}
\def\J{\mathcal J}
\def\K{\mathcal K}
\def\SS{\mathcal S}
\def\T{\mathbb T}
\def\U{\mathcal U}
\def\W{\mathcal W}
\def\one{I}
\def\ie{{\it i.e.}}

\begin{document}

\title{Index theorems for Fredholm, semi-Fredholm, and almost periodic operators: all in one example}

\author{H. Inoue, S. Richard\footnote{Supported by the grant
\emph{Topological invariants through scattering theory and noncommutative geometry} from Nagoya University,
and on leave of absence from Univ.~Lyon, Universit\'e Claude Bernard Lyon
1, CNRS UMR 5208, Institut Camille Jordan, 43 blvd.~du 11 novembre 1918, F-69622
Villeurbanne cedex, France}}

\date{\small}
\maketitle \vspace{-1cm}

\begin{quote}
\emph{
\begin{itemize}
\item[] Graduate school of mathematics, Nagoya University,
Chikusa-ku, Nagoya 464-8602, Japan
\item[] \emph{E-mails:} inoue.hideki124@gmail.com, richard@math.nagoya-u.ac.jp
\end{itemize}
}
\end{quote}

\begin{abstract}
Based on operators borrowed from scattering theory,
several concrete realizations of index theorems are proposed.
The corresponding operators belong to some $C^*$-algebras of pseudo-differential operators
with coefficients which either have limits at
$\pm\infty$, or which are periodic or asymptotically periodic, or which are uniformly almost periodic.
These various situations can be deduced from a single partial isometry which depends on several parameters.
All computations are explicitly performed.
\end{abstract}

\smallskip


\section{Introduction}

Levinson's theorem is a relation between the number of bound states of a quantum mechanical system
and an expression related to the scattering part of that system.
It was originally established by N. Levinson in \cite{Lev} for a Schr\"odinger operator with a spherically symmetric potential,
and has then been developed by numerous researchers on a purely analytical basis.
The interpretation of this relation as an index theorem in scattering theory appeared for the first time
in the conference paper \cite{KRr} or on an explicit model in \cite{KR}.
Ten years later, these investigations on a topological approach of Levinson's theorem have been summarized in the review paper \cite{R}
to which we refer for many concrete examples and earlier references on the subject.

The common feature of all these investigations is that the number of eigenvalues (or point spectrum) of the underlying operator
is finite. At the operator level it means that the key role is played by an isometry which is also a Fredholm operator,
and at the algebraic level it implies that the Toeplitz extensions of $C(\S)$ by the set of compact operators provides the necessary framework.
Even though these ingredients look rather simple, let us stress that the main difficulty in these investigations comes from the analytical part,
namely the affiliation to a certain $C^*$-algebra of an operator defined only by a strong limit.

Now, the aim and the content of the present paper are radically different. Based on some analysis which has already been performed in \cite{DR}
and which is recalled in Section \ref{sec_model}, our starting point is an explicit formula for a partial isometry in $L^2(\R)$ which depends
on several parameters. This isometry is a special instance of a pseudo-differential operator of degree $0$ with values
in some prescribed subalgebra of $C_{\rm b}(\R)$. By tuning the parameters, this algebra is either made of functions which have limits at
$\pm\infty$, or which are periodic or asymptotically periodic, or which are uniformly almost periodic.
Accordingly, this isometry is either Fredholm, or semi-Fredholm, or with both an infinite dimensional kernel and cokernel.
In each of these situations, we provide an index theorem which corresponds to an extension of Levinson's theorem when
an infinite number of eigenvalues is involved.

Independently of their interpretations in the framework of scattering theory, our results can be seen as concrete realizations
of some famous index theorems developed in different contexts.
In short, we provide examples of index theorems in the framework of a comparison algebra \emph{\`a la} Cordes \cite[Chap.~V]{Cordes_LNM},
in the framework of an algebra of asymptotically periodic pseudo-differential operator with a $2$-link ideal chain \cite{CM},
or in the framework of an algebra of almost periodic pseudo-differential operators \cite{CDSS,CMS,Shubin1}.
We also stress that all these incarnations are based on a single formula for a wave operator of a rather simple model.

Let us now be more precise on the model and on the results, see also Section \ref{sec_model} for more details on the model.
The initial system consists in a family of Schr\"odinger operators on the half-line $\R_+$ with a potential of the form $1/r^2$.
A parameter $m\in \C$ with $\Re(m)>-1$ is used for describing the coupling constant for the potential, and a parameter $\kappa\in \C$ is used
for defining the boundary condition at $r=0$.
The study of the corresponding family of closed operators $H_{m,\kappa}$ in $L^2(\R_+)$ has been performed in \cite{DR}.
For the pairs of parameters $(m,\kappa)$ exhibited in \eqref{def_Omega} the operators $H_{m,\kappa}$ are self-adjoint;
we shall stick to this family in the sequel.

Given two operators $H_{m,\kappa}$ and $H_{m',\kappa'}$ the M\o ller wave operators $W_{m,\kappa;m',\kappa'}^\mp$
have been explicitly computed. These operators are partial isometries which intertwine the two initial operators, namely
the following equalities hold:
$$
W^{\mp}_{m,\kappa;m',\kappa'}H_{m',\kappa'}=H_{m,\kappa}W_{m,\kappa;m',\kappa'}^{\mp}.
$$
Note that these operators are at the root of scattering theory, and that they allow
a deep investigation on the operator $H_{m,\kappa}$ when a simpler comparison operator
$H_{m',\kappa'}$ is chosen.
More importantly for us, the kernel of these operators correspond to the subspace spanned by the eigenfunctions of $H_{m',\kappa'}$
while their cokernel coincide with the subspace spanned by the eigenfunctions of the operator $H_{m,\kappa}$.

For this model, the explicit formulas for the wave operators involve the $\Gamma$-function.
For simplicity we shall provide a representation of these operators in $L^2(\R)$ instead of $L^2(\R_+)$.
For that purpose let us introduce the bounded and continuous
function $\Xi_m:\R\to \C$ defined for $\xi\in \R$ by
\begin{equation}\label{eq_Xi}
\Xi_m(\xi):=\e^{i\ln(2)\xi}\frac{\Gamma\big(\frac{m+1+i\xi}{2}\big)}{\Gamma\big(\frac{m+1-i\xi}{2}\big)}.
\end{equation}
Then, we show in Proposition \ref{prop_unit_eq} that the operators $W^{\mp}_{m,\kappa;m',\kappa'}$ are
unitarily equivalent to the following expressions in $L^2(\R)$:
\begin{align}\label{expression in line}
\W^{\mp}_{m,\kappa;m',\kappa'} = & \Xi_\frac{1}{2}(-D)\big(\Xi_m(D)-\varsigma\Xi_{-m}(D)\e^{2mX}\big) \frac{\e^{\mp i\frac{\pi}{2}m}}{1-\varsigma \e^{\mp i\pi m}\e^{2mX}}\notag \\
& \quad \times \frac{\e^{\pm i\frac{\pi}{2}m'}}{1-\varsigma' \e^{\pm i\pi m'}\e^{2m'X}} \big(\Xi_{m'}(-D)-\varsigma'\e^{2m'X}\Xi_{-m'}(-D)\big)\Xi_\frac{1}{2}(D),
\end{align}
where $X$ and $D:=-i\partial_x$ denote the usual self-adjoint operators of multiplication and differentiation in $L^2(\R)$.
The constant $\varsigma$ is defined by $\varsigma\equiv \varsigma(m,\kappa):=\kappa\frac{\Gamma(-m)}{\Gamma(m)}$, and accordingly
$\varsigma':=\kappa'\frac{\Gamma(-m')}{\Gamma(m')}$. Note that for shortness the cases $m,m'$ equal to $0$ are not considered here,
but were treated in \cite{DR}.

The next step consists in studying this partial isometry as a function of the parameters $(m,\kappa)$ and $(m',\kappa')$.
In the recent paper \cite{NPR}, a special case was considered, namely when $H_{m,\kappa}$ is a closed operator (not self-adjoint in general)
with a finite number of eigenvalues, and when $H_{\frac{1}{2},0}$ is chosen for $H_{m',\kappa'}$. This latter operator corresponds
to the Dirichlet Laplacian on $\R_+$, and the resulting operator $W^-_{m,\kappa;\frac{1}{2},0}$ is Fredholm, but not a partial
isometry in general. The interest is this special setting was to extend Levinson's theorem to non-self-adjoint operators
and to complex eigenvalues.
In Section \ref{sec_Fredholm} we briefly recall part of these results but only in the self-adjoint case.
However, we put a special emphasize on the $C^*$-algebra which is similar to some comparison algebras introduced in \cite[Chap.~V]{Cordes_LNM}.
In our case, the quotient of this algebra by its commutator ideal corresponds to $C(\triangle)$,
the algebra of continuous functions on the boundary of a triangle.
In this setting, our index theorem is a relation between the winding number of a function defined on $\triangle$
and the index of the operator $\W^{-}_{m,\kappa;\frac{1}{2},0}$, see Theorem \ref{index theorem 3}.

In Section \ref{sec_semi_I} we concentrate on a set of parameter $(m,\kappa)$ leading to an operator $H_{m,\kappa}$
with an infinite number of eigenvalues, together with the reference operator $H_{\frac{1}{2},0}$. The resulting
partial isometry $\W^{-}_{m,\kappa;\frac{1}{2},0}$ has a trivial kernel but an infinite dimensional cokernel.
Fortunately, the operator $\W^{-}_{m,\kappa;\frac{1}{2},0}$ is invariant under the action of a representation of $\Z$,
or equivalently this operator belongs to a $C^*$-algebra of pseudo-differential operators on $\R$ with periodic coefficients.
This situation corresponds to the situation considered in the seminal paper \cite{Atiyah} where
an index theorem is provided for elliptic operators on a non-compact manifold which are invariant under the action
of a discrete group.
The decomposition with respect to the group corresponds in our setting to the Floquet-Bloch decomposition.
A suitable $\Z$-index can then be defined, and the resulting index theorem
is an equality between this $\Z$-index computed on $\W^{-}_{m,\kappa;\frac{1}{2},0}$ and the
winding number computed on one period of a periodic function.
We refer to Theorem \ref{index theorem 4} for the details.

What happens now if the pseudo-differential operator is no more periodic but only asymptotically periodic\;\!?
This is the situation studied in Section \ref{sec_semi_II}. Indeed, for a suitable family of parameters,
the operator $\W^{-}_{m,\kappa;m',\kappa'}$ has such a property, and correlatively
possesses a finite dimensional kernel and an infinite dimensional cokernel.
In this case the relevant $C^*$-algebra is made of pseudo-differential operators with asymptotically periodic coefficients.
Such an algebra has been studied for example in \cite{CM} and its $K$-theory in \cite{MS},
and its main feature is a $2$-link ideal chain.
As a consequence two index theorems can be deduced. The first one, Theorem \ref{index theorem 6}, is related
to the infinite dimensional cokernel and corresponds to two copies of the index theorem
for pseudo-differential operators with periodic coefficients, as studied in Section \ref{sec_semi_I}.
These two systems are the ones located at $\pm \infty$. However, the drawback of this result is
that the information about the finite dimensional kernel is lost.
On the other hand, this information can be recovered if a fixed reference system is chosen such that
the cokernel is filled. Once this operation done, the resulting operator is simply a partial isometry
with a finite dimensional kernel and a trivial cokernel. The necessary index theorem can now be borrowed
from the Section \ref{sec_Fredholm} for Fredholm operators.
This result is presented in Theorem \ref{thm_plugged}.

In Section \ref{sec_ap} we finally investigate the case when the parameters $(m,\kappa)$ and $(m',\kappa')$
lead to an operator $\W^-_{m,\kappa;m',\kappa'}$ which belongs to an algebra
of pseudo-differential operators with uniformly almost periodic coefficients, as thoroughly studied
in \cite{CDSS,CMS,Shubin1}. This situation also corresponds to an operator $\W^-_{m,\kappa;m',\kappa'}$
with infinite dimensional kernel and cokernel. In this context, suitable generalizations for the notions of
winding number and trace are proposed, and the corresponding index theorem is stated in Theorem
\ref{index theorem 5}.
Let us finally mention that an additional information is available in our context. Indeed, since the kernel and cokernel of
$\W^-_{m,\kappa;m',\kappa'}$ are related to the eigenfunctions of the operators $H_{m',\kappa'}$ and $H_{m,\kappa}$ respectively,
the relative densities of the corresponding eigenvalues can be computed. This computation is presented in Remark \ref{rem_dens}.

As a conclusion, let us confess that the index theorems presented above are certainly not the most general ones. However, they correspond to
incarnations of several famous theorems, and for these results not so many explicitly computable examples exist.
We hope that having concrete examples at hand can be useful for further developing some abstract theory.
For example, the algebra and the corresponding index theorems presented in Section \ref{sec_semi_II} certainly deserve
more investigations, since the results are not so satisfactory.
On the other hand, let us emphasize that these investigations correspond to a fruitful interaction between
scattering theory and operator algebras, two fields of research which are usually not known for their reciprocal emulation.

\section{The model}\label{sec_model}

In this section we introduce the model which leads to the expression \eqref{expression in line}
and provide some information about its spectral and scattering theory.
This material is borrowed from \cite{DR} to which we refer for more explanations and for the proofs.

For any $m \in \C$ we consider the differential expression
\begin{equation*}
L_{m^2}:=-\partial_r^2+\Bigl(m^2-\frac{1}{4}\Bigr)\frac{1}{r^2}
\end{equation*}
acting on distributions on $\R_+$.
The maximal operator associated with $L_{m^2}$ in $L^2(\R_+)$ is defined by
$\D(L_{m^2}^{\max})=\{f\in L^2(\R_+)\mid L_{m^2}f\in L^2(\R_+)\}$, and the minimal operator
$\D(L_{m^2}^{\min})$ is the closure of the restriction of $L_{m^2}$ to $C_{c}^\infty(\R_+)$
(the set of compactly supported smooth functions on $\R_+$).
Then, the equality $(L_{m^2}^{\min})^*=L_{\bar{m}^2}^{\max}$ holds for any $m\in\C$,
and $L_{m^2}^{\min}=L_{m^2}^{\max}$ if $|\Re(m)|\geq 1$ while
$L_{m^2}^{\min}\subsetneq L_{m^2}^{\max}$ if $|\Re(m)|<1$.
In the latter situation $\D(L_{m^2}^{\min})$ is a closed subspace of codimension $2$ of $\D(L_{m^2}^{\max})$,
and for any $f\in\D(L_{m^2}^{\max})$ there exist $a,b\in\C$ such that
\begin{equation*}
f(r)-ar^{\frac{1}{2}-m}-br^{\frac{1}{2}+m}\in \D(L_{m^2}^{\min})\ \text{around}\ 0.
\end{equation*}
Here, the expression $g(r)\in\D(L_{m^2}^{\min})$ \emph{around} $0$ means that there exists
$\zeta\in\ C_c^{\infty}\big([0,\infty)\big)$ with $\zeta=1$ around $0$ such that $g\zeta\in\D(L_{m^2}^{\min})$.
It is worth mentioning that the functions $r \mapsto r^{\frac{1}{2}\pm m}$ are the two linearly independent solutions
of the ordinary differential equation $L_{m^2}u=0$, and that they are square integrable near $0$ if $|\Re(m)|<1$.

Based on the above observations we construct various closed extensions of the operator $L_{m^2}^{\min}$.
For simplicity we restrict our attention to $m\in \C$ with $|\Re(m)|<1$. These extensions
are parameterized by a boundary condition at $0$, namely for any $\kappa\in\C$
we define the operator $H_{m,\kappa}$ by considering the restriction of the operator $L_{m^2}^{\max}$ on the domain
\begin{align*}
\D(H_{m,\kappa})=\big\{&f\in\D(L_{m^2}^{\max})\mid \text{for some}\ c\in\C,\\
&\qquad \qquad \qquad f(r)-c(\kappa r^{\frac{1}{2}-m}+r^{\frac{1}{2}+m})\in\D(L_{m^2}^{\min})\ \text{around}\ 0\big\}.
\end{align*}
Note that if $m=0$ these operators do not depend on $\kappa$, and in fact another natural family of extensions exist in this case.
For simplicity we shall not consider this family here, and therefore exclude the case $m=0$.
Note also that if $m=\frac{1}{2}$ and $\kappa=0$, the operator $H_{\frac{1}{2},0}$ corresponds to the \emph{Dirichlet Laplacian} on $\R_+$,
which will be denoted by $H_{\rm D}$. Let us also mention that $H_{m,\kappa}$ is self-adjoint if
$(m,\kappa)$ belongs to the set
\begin{equation}\label{def_Omega}
\Omega_{\rm sa}:=\big\{(m,\kappa)\mid m\in (-1,1) \setminus \{0\}\hbox{ and } \kappa \in \R, \hbox{ or } m\in i \R^* \hbox{ and } |\kappa|=1\big\}.
\end{equation}
In the sequel, we shall stick to the self-adjoint case, which means we shall always consider $(m,\kappa)$ in the above set.

In the next statement we provide some information on the point spectrum $\sigma_{\rm p}(H_{m,\kappa})$ of $H_{m,\kappa}$.
For $(m,\kappa)$ in $\Omega_{\rm sa}$ we set
$$
\varsigma=\varsigma(m,\kappa):=\kappa\frac{\Gamma(-m)}{\Gamma(m)}
$$
with $\Gamma$ the usual Gamma function.

\begin{theorem}[\cite{DR}, Theorem 5.2 \& 5.5]\label{spectral theory}
For $(m,\kappa)$ in $\Omega_{\rm sa}$ one has
\begin{align*}
\sigma_{\rm p}(H_{m,\kappa})=\begin{cases}
\left\{-4\e^{-w}\mid w\in \frac{1}{m}\Ln(\varsigma)\ \text{and}\ \Im(w)\in(-\pi,\pi)\right\} & \text{if}\ \kappa \neq0\\
\emptyset & \text{if}\  \kappa=0,
\end{cases}
\end{align*}
where $\Ln$ denotes the set of values of the multivalued logarithm. In particular,
\begin{enumerate}
\item[(i)] If $m\in (-1,1)\setminus\{0\}$ and $\kappa\in \R$ then $H_{m,\kappa}$ has a finite number of eigenvalues,
\item[(ii)] If $m=in$ with $n \in \R^*$ and if $|\kappa|=1$, then
\begin{equation*}
\sigma_p(H_{in,\kappa})=\left\{-4\exp\Bigl(-\frac{\arg(\varsigma)+2\pi j}{n}\Bigr)\mid j\in\Z\right\}.
\end{equation*}
\end{enumerate}
\end{theorem}

We now review one construction related to the scattering theory for the operators $H_{m,\kappa}$.
More generally, this theory can be elaborated in a time-dependent framework or in a time-independent framework.
Here we stick to the second formulation and provide only a minimal set of information,
see \cite[Sec.~6]{DR} for a more general picture.

For any $(m,\kappa)$ in $\Omega_{\rm sa}$ we define the \emph{incoming} and \emph{outgoing Hankel transforms}
$\F^{\mp}_{m,\kappa}: C_c^{\infty}(\R_+)\to L^2(\R_+)$ by the integral kernels
\begin{equation*}
\F_{m,\kappa}^{\mp}(r,s):=\e^{\mp i\frac{\pi}{2}m}\sqrt{\frac{2}{\pi}}\frac{\J_m(rs)-
\varsigma\J_{-m}(rs)\left(\frac{s}{2}\right)^{2m}}{1-\varsigma \e^{\mp i\pi m}\left(\frac{s}{2}\right)^{2m}}\qquad \text{for}\ r,s\in\R_+
\end{equation*}
and where $\J_m(r):=\sqrt{\frac{\pi r}{2}}J_m(r)$ with $J_m$ the Bessel function.
These operators extend continuously to bounded operators on $L^2(\R_+)$.
If  $(m',\kappa')$ also belong to $\Omega_{\rm sa}$ we then set
\begin{equation*}
W_{m,\kappa;m',\kappa'}^{\mp}:=\F^{\mp}_{m,\kappa}\F_{m',\kappa'}^{\mp *}\ ,
\end{equation*}
with ${}^*$ denoting the adjoint operator.
These operators correspond to the so-called M\o ller wave operators and play a major role in scattering theory.
Some properties of the Hankel transformations and of the wave operators are gathered in the next statement.

\begin{proposition}\label{properties}
Let $(m,\kappa)$ and $(m',\kappa')$ belong to $\Omega_{\rm sa}$ . Then,
\begin{enumerate}
\item[(i)] $\F_{m,\kappa}^{\mp *}\F_{m,\kappa}^{\mp}=\one$,
\item[(ii)] $\one_{\rm c}(H_{m,\kappa}):=\F_{m,\kappa}^{\mp}\F_{m,\kappa}^{\mp *}$ corresponds to
the projection onto the continuous subspace of $H_{m,\kappa}$,
\item[(iii)] $W_{m,\kappa;m',\kappa'}^{\mp *}W_{m,\kappa;m',\kappa'}^{\mp}=\one_{\rm c}(H_{m',\kappa'})$,
\item[(iv)] $W_{m,\kappa;m',\kappa'}^{\mp}W_{m,\kappa;m',\kappa'}^{\mp *}=\one_{\rm c}(H_{m,\kappa})$,
\item[(v)] The intertwining relation $W^{\mp}_{m,\kappa;m',\kappa'}H_{m',\kappa'}=H_{m,\kappa}W_{m,\kappa;m',\kappa'}^{\mp}$ holds.
\end{enumerate}
\end{proposition}

In addition to the previous relations, if $\one_{\rm p}(H_{m,\kappa})$ denotes the orthogonal projection onto the subspace
generated by the eigenfunctions of $H_{m,\kappa}$, then one has
\begin{equation}\label{projections}
\big[W_{m,\kappa;m',\kappa'}^{\mp},W_{m,\kappa;m',\kappa'}^{\mp *}\big]=\one_{\rm p}(H_{m',\kappa'})-\one_{\rm p}(H_{m,\kappa}),
\end{equation}
where we have used the standard notation for the commutator.

Our main interest in the model introduced so far is that explicit formulas for the
wave operators can be exhibited.
For that purpose let us introduce \emph{the dilation group} $\{U_t\}_{t\in\R}$ on $L^2(\R_+)$ defined by
$[U_tf](r):=\e^{t/2}f(\e^tr)$ for any $f\in L^2(\R_+)$, $t\in \R$ and $r\in\R_+$.
Its self-adjoint generator $A$ is called \emph{the generator of dilations}
and corresponds to the differential expression
\begin{equation*}
A=\frac{1}{2i}(R\partial_r+\partial_rR),
\end{equation*}
where $R$ denotes the multiplication operator by the variable in $\R_+$,
\ie~$[Rf](r)=rf(r)$ for $f$ in a suitable subset of $L^2(\R_+)$.

In the following statement we recall the expressions for the wave operators
obtained in \cite[Sec.~6]{RS}. Before this, let us recall that the bounded and continuous
function $\Xi_m:\R\to \C$ has been introduced in \eqref{eq_Xi}, and let us set $J$
for the self-adjoint and unitary operator in $L^2(\R_+)$ defined by $[Jf](r):=\frac{1}{r}f\left(\frac{1}{r}\right)$.

\begin{lemma}\label{expression in half-line}
Let $(m,\kappa)$ and $(m',\kappa')$ belong to $\Omega_{\rm sa}$, and set
$\varsigma:=\kappa\frac{\Gamma(-m)}{\Gamma(m)}$ and $\varsigma':=\kappa'\frac{\Gamma(-m')}{\Gamma(m')}$.
Then the operators $W_{m,\kappa;m',\kappa'}^{\mp}$ coincide with the operators
\begin{align*}
& J\Big(\Xi_m(A)-\varsigma\Xi_{-m}(A)\big(\tfrac{R}{2}\big)^{2m}\Big)
\frac{\e^{\mp i\frac{\pi}{2}m}}{1-\varsigma \e^{\mp i\pi m}\left(\frac{R}{2}\right)^{2m}}\\
&\quad \times \frac{\e^{\pm i\frac{\pi}{2}m'}}{1-\varsigma' \e^{\pm i\pi m'}\left(\frac{R}{2}\right)^{2m'}}
\Big(\Xi_{m'}(-A)-\varsigma'\big(\tfrac{R}{2}\big)^{2m'}\Xi_{-m'}(-A)\Big)J.
\end{align*}
\end{lemma}

From the previous statement we can now deduce the representation of the wave operators in $L^2(\R)$
provided in \eqref{expression in line}. The proof is based on Stone-von Neumann theorem.
In its statement and in the sequel we shall use the $C^*$-algebra $C\big([-\infty,\infty]\big)$ of continuous
functions having limits at $\pm \infty$.
We also denote by $C_{\rm b}(\R)$ the $C^*$-algebra of bounded and continuous
functions on $\R$.
For the following proof, we shall use the property which we recall from the proof
of \cite[Thm.~4.10]{DR}: For any $m,m'\in \C$ with $\Re(m)>-1$ and $\Re(m')>-1$
the map $\xi\mapsto \Xi_m(-\xi)\Xi_{m'}(\xi)$ belongs to $C\big([-\infty,\infty]\big)$
and the following equalities hold:
\begin{equation}\label{eq_limits}
\Xi_m(\mp \infty)\Xi_{m'}(\pm \infty)=\e^{\mp i \frac{\pi}{2}(m-m')}.
\end{equation}

\begin{proposition}\label{prop_unit_eq}
Let $(m,\kappa),$ $(m',\kappa')$, $\varsigma$ and $\varsigma'$ be as in Lemma \ref{expression in half-line}.
Then, the wave operators $W_{m,\kappa;m',\kappa'}^{\mp}$ are unitarily equivalent to the operators
provided in \eqref{expression in line} and acting in $L^2(\R)$.
Moreover, these operators belong to the unital $C^*$-subalgebra $\E$ of $\B\big(L^2(\R)\big)$ generated by
products of the form $a(D)b(X)$ with $a\in C\big([-\infty,\infty]\big)$ and $b\in C_{\rm b}(\R)$.
\end{proposition}

\begin{proof}
Let us set $B:=\frac{1}{2}\ln\big(\frac{R^2}{4}\big)$. Then, by a straightforward computation one infers
that for any $s,t\in\R$ the usual Weyl commutation relation holds, namely
\begin{equation*}
\e^{isB}\e^{itA}=\e^{-ist}\e^{itA}\e^{isB}.
\end{equation*}
It then follows from Stone-von Neumann theorem that there exists a unitary operator
$V:L^2(\R_+)\to L^2(\R)$ such that $VB=XV$ and $VA=DV$.
Now, let us recall that $J$ is a self-adjoint and unitary operator in $L^2(\R_+)$.
In addition, one easily infers from \eqref{eq_Xi} that $\Xi_{\frac{1}{2}}(A)$ is a unitary operator which satisfies
$\Xi_{\frac{1}{2}}(A)^* =\Xi_{\frac{1}{2}}(-A)$.
One then deduces that the operator
$V \Xi_{\frac{1}{2}}(-A) J:L^2(\R_+)\to L^2(\R)$ realizes the unitary equivalence between $W_{m,\kappa;m',\kappa'}^{\mp}$
and the operators provided in \eqref{expression in line}.
Note that a similar but more explicit transformation is also provided in \cite[App.~B]{DR}.

For the second part of the statement, by taking \eqref{eq_limits} into account,
one readily observes that the symbols of the operator given in \eqref{expression in line} belong to
$C\big([-\infty,\infty]\big)$ and to $C_{\rm b}(\R)$, respectively.
\end{proof}

In the sequel the representation \eqref{expression in line} in $L^2(\R)$ of the wave operator will constantly be used.
For that reason we denote these operators by the very similar notation $\W^{\mp}_{m,\kappa;m',\kappa'}$.
The projections on their kernel and cokernel will be denoted by $\I_{\rm p}(H_{m',\kappa'})$ and
$\I_{\rm p}(H_{m,\kappa})$, respectively. This notation is in accordance with \eqref{projections}.

Observe now that the $C^*$-algebra $\E$ in the above theorem has a non-trivial ideal $\J$
generated by $a(D)b(X)$ with $a\in C_0(\R)$ and $b\in C_{\rm b}(\R)$. Here $C_0(\R)$
denotes the algebra of continuous functions on $\R$ converging to $0$ at $\pm \infty$.
Then, one has $\E/\J\cong C_{\rm b}(\R)\oplus C_{\rm b}(\R)$ and the corresponding quotient map $q_{{}_\J}$ is
given by the evaluation of Fourier multipliers at $\pm\infty$, \ie~$a(D)\mapsto \big(a(-\infty),a(+\infty)\big)$.
By a direct computation, we obtain
\begin{equation}\label{eq_2_limits}
q_{{}_\J}\big(\W_{m,\kappa;m'\kappa'}^{-}\big)=\big(\SS_{m,\kappa;m',\kappa'},1\big)
\end{equation}
with
\begin{equation}\label{eq_S}
\SS_{m,\kappa;m',\kappa'}(x):=\e^{-i\pi(m-m')}\frac{1-\varsigma \e^{ i\pi m}\e^{2mx}}{1-\varsigma \e^{- i\pi m}\e^{2mx}}
\ \frac{1-\varsigma' \e^{- i\pi m'}\e^{2m'x}}{1-\varsigma' \e^{ i\pi m'}\e^{2m'x}}\ .
\end{equation}
This operator is in fact unitarily equivalent to the scattering operator for the pair of operators $(H_{m,\kappa},H_{m',\kappa'})$.
Indeed, by definition the scattering operator $S_{m,\kappa;m',\kappa'}$ is given by the product $W_{m,\kappa;m'\kappa'}^{+*}W_{m,\kappa;m'\kappa'}^{-}$
and is a unitary operator in $\one_{\rm c}(H_{m',\kappa'})L^2(\R_+)$.
By taking the points (i) and (ii) of Proposition \ref{properties} into account one infers
that the operator $\F_{m',\kappa'}^{+ *}S_{m,\kappa;m',\kappa'}\F_{m',\kappa'}^+$ is a unitary operator in $L^2(\R_+)$,
and as shown in \cite[Sec.~6.5]{DR} one has
\begin{equation}\label{eq_true_S}
\F_{m',\kappa'}^{+ *}S_{m,\kappa;m',\kappa'}\F_{m',\kappa'}^+
= \e^{- i \pi(m-m')}\;\!\frac{1-\varsigma\e^{i\pi m}(\frac{R}{2})^{2m}}{1-\varsigma\e^{-i\pi m}(\frac{R}{2})^{2m}}
\ \frac{1-\varsigma'\e^{-i\pi m'}(\frac{R}{2})^{2m'}}{1-\varsigma'\e^{i\pi m'}(\frac{R}{2})^{2m'}}.
\end{equation}
Finally, by using the operator $B$ and the unitary transformation $V$ introduced in the proof of Proposition \ref{prop_unit_eq}
one obtains that \eqref{eq_true_S} is unitarily equivalent to \eqref{eq_S}.

It turns out that the algebra $\E$ introduced in Proposition \ref{prop_unit_eq} is too large for formulating an index theorem.
In the following sections we construct more suitable algebras which depend
on the choice of parameters $(m,\kappa)$ and $(m',\kappa')$.
Also, since we shall consider only the operator $\W_{m,\kappa;m'\kappa'}^{-}$ and not $\W_{m,\kappa;m'\kappa'}^{+}$,
we shall choose an algebra in which the superfluous $1$ in \eqref{eq_2_limits}
does not appear any more.

\section{The Fredholm case}\label{sec_Fredholm}

This situation is the simplest one and has already been thoroughly investigated in \cite{NPR},
with a special emphasize on the non self-adjoint cases.
Here we briefly introduce a slightly updated version the main result in the self-adjoint setting.

For this case we consider the pairs of parameters $(m,\kappa)$ and $(\frac{1}{2},0)$ with
$m\in (0,1)$ and $\kappa \in \R$.
Note that there is no loss of generality in considering $m>0$ since the relation $H_{-m,\kappa}=H_{m,\kappa^{-1}}$ holds
(in \cite{DR} an operator $H_{m,\infty}$ is also introduced but it has not been recalled here for simplicity).
For these parameters the operator $H_{m,\kappa}$ is self-adjoint with a finite number of eigenvalues of finite multiplicity,
and $H_{\frac{1}{2},0}\equiv H_{\rm D}$ corresponds to the Dirichlet Laplacian, with no eigenvalue.
The corresponding wave operator $\W^-_{m,\kappa;\frac{1}{2},0}$ is then an isometry with a finite dimensional cokernel,
and is thus a Fredholm operator.
Its analytic index is given by minus the number of eigenvalues of the operator $H_{m,\kappa}$.

For the pair of operators $(H_{m,\kappa},H_{\rm D})$ we consider the unital $C^*$-subalgebra $\E_{(X,D)}$ of $\B\big(L^2(\R)\big)$
defined by
$$
\E_{(X,D)}:=C^*\Big(a(D)b(X)\mid a\in C_0\big([-\infty,\infty)\big), b\in C\big([-\infty,+\infty]\big) \Big)^+.
$$
Here $C_0\big([-\infty,\infty)\big)$ denotes the set of continuous functions on $\R$ having a limit at $-\infty$
and vanishing at $+\infty$.
By observing then that the ideal of compact operators $\K_\R:=\K\big(L^2(\R)\big)$ corresponds to the $C^*$-algebra generated
by products of the form $a(D)b(X)$ with $a,b\in C_0(\R)$,
one easily infers that  $\E_{(X,D)}/\K_{\R}$ is isomorphic to $C(\triangle)$,
the algebra of continuous functions on the boundary $\triangle$
of the closed triangle $\blacktriangle$.
Note that the unital quotient morphism $q_{{}_{\K_\R}}:\E_{(X,D)}\to C(\triangle)$ is uniquely determined by
$$
q_{{}_{\K_\R}}\big(a(D)b(X)\big) = \big(a(\cdot)b(-\infty),\;\! a(-\infty)b(\cdot),\;\!a(\cdot)b(+\infty)\big).
$$

The next step consists in observing that for $m\in (0,1)$ and $\kappa\in \R$ the operator $\W_{m,\kappa;\frac{1}{2},0}^{-}$
belongs to $\E_{(X,D)}$.
In addition, by a direct computation one gets for $\kappa\neq 0$
\begin{equation}\label{eq_tri}
q_{{}_{\K_\R}}\big(\W_{m,\kappa;\frac{1}{2},0}^{-}\big)
= \Big(\e^{i\frac{\pi}{2}(\frac{1}{2}-m)}\Xi_{\frac{1}{2}}(-\cdot)\Xi_m(\cdot),\;\! \SS_{m,\kappa;\frac{1}{2},0},
\;\! \e^{i\frac{\pi}{2}(\frac{1}{2}+m)}\Xi_{\frac{1}{2}}(-\cdot)\Xi_{-m}(\cdot) \Big)
\end{equation}
with $\SS_{m,\kappa;\frac{1}{2},0}(x) = \e^{-i\pi(m-\frac{1}{2})}\frac{(1-\varsigma \e^{ i\pi m}\e^{2mx})}{(1-\varsigma \e^{- i\pi m}\e^{2mx})}$
introduced in \eqref{eq_S}.
In the special case $\kappa =0$ one has
\begin{equation}\label{eq_tri_2}
q_{{}_{\K_\R}}\big(\W_{m,0;\frac{1}{2},0}^{-}\big)
= \Big(\e^{i\frac{\pi}{2}(\frac{1}{2}-m)}\Xi_{\frac{1}{2}}(-\cdot)\Xi_m(\cdot),\;\! \SS_{m,0;\frac{1}{2},0},
\;\! \e^{i\frac{\pi}{2}(\frac{1}{2}-m)}\Xi_{\frac{1}{2}}(-\cdot)\Xi_{m}(\cdot) \Big).
\end{equation}
For simplicity we denote by $\Gamma_{m,\kappa;\frac{1}{2},0}^{\triangle}$ the r.h.s.~of \eqref{eq_tri}
and $\Gamma_{m,0;\frac{1}{2},0}^{\triangle}$ the r.h.s.~of \eqref{eq_tri_2}.
Since the operator is $\W_{m,\kappa;\frac{1}{2},0}^{-}$ is Fredholm, it follows that
$\Gamma_{m,\kappa;\frac{1}{2},0}^{\triangle}$ is a function taking values in $\T:=\{z\in \C\mid |z|=1\}$.

We finally observe that $\triangle$ is homeomorphic to the $1$-circle $\S$. We can then define the winding number
of any invertible function $f\in C(\triangle)$, which we denote by $\wn_{\triangle}(f)$.
By convention we compute the winding number by turning anticlockwise around $\S$ and the increase
in the winding number is also counted anticlockwise.
Rephrasing then for the self-adjoint operators $(H_{m,\kappa},H_{\rm D})$ the content of \cite[Thm.~2]{NPR}
in the framework introduced above, we get:

\begin{theorem}\label{index theorem 3}
For any $m\in (0,1)$ and for any $\kappa\in \R$ one has
\begin{equation}\label{eq_Lev}
\wn_{\triangle}\big(\Gamma_{m,\kappa;\frac{1}{2},0}^{\triangle}\big)=
\text{number of eigenvalues of}\ H_{m,\kappa}=
-\Index\big(\W_{m,\kappa;\frac{1}{2},0}^{-}\big).
\end{equation}
\end{theorem}

\begin{remark}
The previous statement corresponds to what is called \emph{a topological version of Levinson's theorem}, see \cite{R}.
Indeed, the function on the lower edge of the triangle $\triangle$ is given by the scattering
operator introduced in \eqref{eq_S} and the functions on the other edges define corrections at threshold energies.
Thus, equality \eqref{eq_Lev} corresponds to a relation between the scattering operator
and the number of eigenvalues of $H_{m,\kappa}$, with some natural corrections taken into account.
Note that this type of relations first appeared in \cite{Lev} and had been widely popularized since then.
\end{remark}

\section{The semi-Fredholm case (I)}\label{sec_semi_I}

In this section we consider the pair of operators $(H_{in,\kappa},H_{\rm D})$ with $n\in\R^*$ and $\kappa\in\T$.
In fact, by taking the equality $H_{in,\kappa}=H_{-in,\kappa^{-1}}$ into account, one can focuss without loss of generality
on the case $n>0$. For such a pair of operators the wave operator $\W^-_{in,\kappa;\frac{1}{2},0}$ is an isometry with an
infinite dimensional cokernel $\I_{\rm p}(H_{in,\kappa})L^2(\R)$.

In the present situation the key observation based on \eqref{expression in line} is that the following equality holds
\begin{equation*}
\W^-_{in,\kappa;\frac{1}{2},0} = \e^{ i\frac{\pi}{4}} \Xi_\frac{1}{2}(-D)\left(\Xi_{in}(D)-\varsigma\Xi_{-in}(D)\e^{2inX}\right)
\frac{\e^{\frac{\pi}{2}n}}{1-\varsigma \e^{\pi n}\e^{2inX}}
\end{equation*}
and that the corresponding symbol is periodic in $x$. Hence, a suitable $C^*$-algebra $\E_n$ for nesting this operator
is the unitization of the $C^*$-subalgebra generated by $a(D)b(X)$ with $a\in C_0\big([-\infty,\infty)\big)$
and $b\in C_{\frac{\pi}{n}}(\R)$, namely
\begin{equation}\label{eq_En}
\E_n:=C^*\Big(a(D)b(X)\mid a\in C_0\big([-\infty,+\infty)\big),b\in C_{\frac{\pi}{n}}(\R) \Big)^+.
\end{equation}
Here $C_{\frac{\pi}{n}}(\R)$ denotes the set of all continuous periodic functions on $\R$ with period $\frac{\pi}{n}$.
Clearly, the algebra $\E_n$ is the unitization of an extension of $C_{\frac{\pi}{n}}(\R)$ by the ideal
$\J_n$ defined by
\begin{equation}\label{eq_Jn}
\J_n:=C^*\Big(a(D)b(X)\mid a\in C_0(\R),b\in C_{\frac{\pi}{n}}(\R) \Big).
\end{equation}
Since $C_{\frac{\pi}{n}}(\R)$ can naturally be identified with $C(\S)$,
we define through this identification the winding number  $\wn_{\frac{\pi}{n}}(f)$ of any invertible element
$f\in C_{\frac{\pi}{n}}(\R)$, with the conventions introduced in the previous section.

In order to define an analytic index for $\W^-_{in,\kappa;\frac{1}{2},0}$ we recall a direct integral decomposition of $L^2(\R)$
useful for periodic systems, the so-called \emph{Floquet-Bloch decomposition} introduced for example in \cite[Sec.XIII.16]{RS}.
For each $\theta\in [0,2n)$ we set $\H_\theta:=L^2\big([0,\frac{\pi}{n}], \d x\big)$ endowed with the usual Lebesgue measure, and also define
\begin{equation*}
\H_n:=\int_{[0,2n)}^{\oplus}\H_\theta \frac{\d\theta}{2n}.
\end{equation*}
Then, if $\SS(\R)$ denotes the Schwartz space on $\R$, the map $\U_n:L^2(\R)\to\H_n$ defined for $\theta\in [0,2n)$ and $x\in[0,\frac{\pi}{n})$ by
\begin{equation*}
[\U_nf](\theta,x):=\sum_{k\in\Z}\e^{-i\frac{\pi}{n}k\theta}f\Big(x+\frac{\pi}{n}k\Big) \qquad \forall f \in \SS(\R),
\end{equation*}
extends continuously to a unitary operator.
The adjoint operator is then given by the formula
$$
[\U_n^*\varphi]\big(x+\frac{\pi}{n}k\big) = \int_0^{2n}\e^{i\frac{\pi}{n}k\theta}\varphi(\theta,x)\frac{\d \theta}{2n}.
$$
Moreover, one has
\begin{equation*}
\U_n D\U_n^*=\int_{[0,2n)}^{\oplus} D_x^{(\theta)} \frac{\d\theta}{2n},
\end{equation*}
where $D_x^{(\theta)}$ is the operator $-i\frac{\d}{\d x}$ on a fiber $\H_\theta$
with boundary condition $f(\frac{\pi}{n})=\e^{i\frac{\pi}{n} \theta}f(0)$.

Thus, for any operator of the form $a(D)b(X)$ with $a\in C_0\big([-\infty,\infty)\big)$
and $b\in C_{\frac{\pi}{n}}(\R)$, the operator $\U_na(D)b(X)\U_n^*$ is a decomposable operator with the fibers $a\big(D_x^{(\theta)}\big)b(X)$.
On suitable bounded decomposable operator $\Phi=\int_{[0,2n)}^{\oplus}\Phi(\theta) \frac{\d\theta}{2n}$
we also define the trace $\Trace_n$ by
\begin{equation*}
\Trace_n(\Phi)=\int_{0}^{2n}\trace_\theta\big(\Phi(\theta)\big) \frac{\d\theta}{2n}
\end{equation*}
where $\trace_\theta$ is the usual trace on $\H_\theta$.
The main result for the semi-Fredholm operator $\W_{in,\kappa;\frac{1}{2},0}^-$ reads:

\begin{theorem}\label{index theorem 4}
Let $n>0$ and $\kappa\in \T$. Then,
\begin{equation*}
\wn_{\frac{\pi}{n}}\big(\SS_{in,\kappa;\frac{1}{2},0}\big)=- 1
=\Trace_n \big(\big[\W_{in,\kappa;\frac{1}{2},0}^-,\W_{in,\kappa;\frac{1}{2},0}^{-*}\big]\big).
\end{equation*}
\end{theorem}

\begin{proof}
We directly infer from the equalities
\begin{equation*}
\SS_{in,\kappa;\frac{1}{2},0}(x)=i\e^{\pi n}\frac{1-\varsigma \e^{ -\pi n}\e^{2inx}}{1-\varsigma \e^{ \pi n}\e^{2inx}}
=-i\overline{\varsigma}\e^{-2inx}\;\! \left(\frac{1-\varsigma \e^{-\pi n}\e^{2in x}}{|1-\varsigma \e^{-\pi n}\e^{2in x}|}\right)^2
\end{equation*}
that the winding number $\wn_{\frac{\pi}{n}}\big(\SS_{in,\kappa;\frac{1}{2},0}\big)$ is $-1$.

Let us now consider an operator of the form $a(D)b(X)$ with $a\in C_0(\R)$
and $b\in C_{\frac{\pi}{n}}(\R)$, and the corresponding operator $a\big(D_x^{(\theta)}\big)b(X)$.
Since the eigenfunctions of the operator $D_x^{(\theta)}$ are provided by the functions
$$
[0,\frac{\pi}{n})\ni x \mapsto \sqrt{{\frac{n}{\pi}}}\e^{i(\theta + 2nk)x}\in \C, \qquad k \in \Z
$$
we infer that the Schwartz kernel of the operator $a\big(D_x^{(\theta)}\big)b(X)$ is given by
\begin{equation*}
K_{a(D_x^{(\theta)})b(X)}(x,y)=\frac{n}{\pi}\sum_{k\in\Z}a(\theta+2nk) \e^{i(\theta+2nk)(x-y)} b(y).
\end{equation*}
Thus, if $b(X)a\big(D_x^{(\theta)}\big)$ is
$\trace_\theta$-trace class for a.e.~$\theta \in [0,2n)$ we obtain
\begin{equation*}
\trace_\theta\big(a\big(D_x^{(\theta)}\big)b(X)\big)
=\sum_{k\in\Z}a(\theta+2nk)\times \frac{n}{\pi}\int_{0}^{\frac{\pi}{n}}b(x)\;\!\d x
\end{equation*}
and then
\begin{equation}\label{formula of tracen}
\Trace_n\big(a(D)b(X)\big)=\frac{1}{2n}\int_{\R}a(\xi)\;\!\d\xi\times \frac{n}{\pi}\int_{0}^{\frac{\pi}{n}}b(x)\;\!\d x.
\end{equation}
Note that these formulas are valid if $a$ has a fast enough decay, which will be the case in the sequel.
In addition, note also that the last term depends only on the $0$-th Fourier coefficient of the function $b$.
Our next aim is thus to show that $[\W_{in,\kappa;\frac{1}{2},0}^-,\W_{in,\kappa;\frac{1}{2},0}^{-*}] = -\I_{\rm p}(H_{in,\kappa})$
can be rewritten in the above form.

Recall now that $\Xi_{\frac{1}{2}}(D)$ is a unitary operator, with $\Xi_{\frac{1}{2}}(D)^*=\Xi_{\frac{1}{2}}(-D)$.
One also infers from the definition in \eqref{eq_Xi} that $\Xi_{in}(D)$ is invertible with $\Xi_{in}(D)^{-1} = \Xi_{in}(-D)$
and that $\Xi_{in}(D)^* = \Xi_{-in}(-D)$.
Then, by taking the statement (iii) of Proposition \ref{properties} into account one gets
\begin{align*}
& \Xi_{in}(D)^{-1} \Xi_{\frac{1}{2}}(D)\Big(\W_{in,\kappa;\frac{1}{2},0}^- \W_{in,\kappa;\frac{1}{2},0}^{-*}
- \W_{in,\kappa;\frac{1}{2},0}^{-*} \W_{in,\kappa;\frac{1}{2},0}^-\Big)\Xi_{\frac{1}{2}}(D)^* \Xi_{in}(D)\\
& = \big(\one-\varsigma G^+_n(D)\e^{2inX}\big)F_{in,\kappa}(X)\big(\varsigma \e^{2inX}G^-_n(D)-\one\big)-\one,
\end{align*}
where
\begin{equation*}
F_{in,\kappa}(X):=\frac{-1}{(1 -\varsigma \e^{- \pi n}\e^{2inX})(1-\varsigma \e^{\pi n}\e^{2inX})}
\end{equation*}
and
\begin{equation*}
G_n^{\pm}(\xi):=\Xi_{\pm in}(-\xi)\Xi_{\mp in}(\xi).
\end{equation*}
From the identity $\Gamma(z+\frac{1}{2})\Gamma(-z+\frac{1}{2})=\frac{\pi}{\cos(\pi z)}$ one then infers that
$$
G_n^{\pm}(\xi)=\e^{\pm \pi n}\frac{\e^{\pi \xi}+\e^{\mp \pi n}}{\e^{\pi \xi}+\e^{\pm \pi n}},
$$
and by taking into account the identity \cite[(6.10)]{DR} written in our framework, namely
\begin{equation*}
\e^{2inX}G_n^-(D)+G_n^+(D)\e^{2inX}=2\cosh(\pi n)\e^{2inX},
\end{equation*}
one then gets
\begin{align*}
& \big(\one-\varsigma G^+_n(D)\e^{2inX}\big)F_{in,\kappa}(X)\big(\varsigma \e^{2inX}G^-_n(D)-\one\big)-\one \\
& = -  F_{in,\kappa}(X) - \varsigma^2G_n^+(D)\e^{2inX}F_{in,\kappa}(X)\e^{2inX}G_n^-(D) \\
&\quad +\varsigma F_{in,\kappa}(X)\e^{2inX}G_n^-(D) + \varsigma G_n^+(D)\e^{2in X}F_{in,\kappa}(X) -\one\\
&=- F_{in,\kappa}(X)-\varsigma^2\e^{2inX}F_{in,\kappa}(X)\e^{2inX}-\varsigma^2G_n^+(D)\big[\e^{2inX}F_{in,\kappa}(X)\e^{2inX},G_n^-(D)\big] \\
&\quad + \varsigma F_{in,\kappa}(X)\big(2\cosh(\pi n) \e^{2inX} -G_n^+(D)\e^{2in X}\big)+ \varsigma G_n^+(D)\e^{2in X}F_{in,\kappa}(X) -\one\\
&=-F_{in,\kappa}(X)\big(1- 2 \varsigma \cosh (\pi n) \e^{2inX} + \varsigma^2\e^{4inX}\big) -\one\\
&\quad -\varsigma^2G_n^+(D)\big[\e^{2inX}F_{in,\kappa}(X)\e^{2inX},G_n^-(D)\big]  - \varsigma \big[F_{in,\kappa}(X),G_n^+(D)\big]\e^{2in X}\\
&= - \underbrace{ \varsigma^2 G_n^+(D)\big[\e^{2inX}F_{in,\kappa}(X)\e^{2inX},G_n^-(D)\big]}_{=:I_{in,\kappa}}
- \underbrace{\varsigma \big[F_{in,\kappa}(X),G_n^+(D)\big]\e^{2in X}}_{=:J_{in,\kappa}}
\end{align*}

For the last step, observe that since the function $F_{in,\kappa}$ is a smooth $\frac{\pi}{n}$-periodic function,
its Fourier series converges uniformly. We can thus write $F_{in,\kappa}(X)=\sum_{\ell\in\Z}c_\ell \e^{2in\ell X}$.
Using the relation $\e^{isX}g(D)\e^{-isX}=g(D-s)$, which holds for any $g\in C_{\rm b}(\R)$ and $s\in\R$, we obtain
\begin{align*}
I_{in,\kappa}&=\varsigma^2\sum_{\ell\in\Z}c_\ell G_n^+(D)\Bigl\{G_n^-\bigl(D-2n(\ell+2)\bigr)-G_n^-(D)\Bigr\} \e^{2in(\ell+2)X}, \\
J_{in,\kappa}&=\varsigma\sum_{\ell\in\Z}c_\ell \Bigl\{G_n^+(D-2n\ell)-G_n^+(D)\Bigr\}  \e^{2in(\ell+1)X}.
\end{align*}
By applying then formula \eqref{formula of tracen} one infers that
$\Trace_n(I_{in,\kappa})=0$ since the $0$-th Fourier coefficient of the corresponding function $b$ is obtained for $\ell=-2$, but the first factor vanishes
precisely when $\ell=-2$.
On the other hand one has
\begin{align*}
\Trace_n(J_{in,\kappa})&=\varsigma c_{-1}\frac{1}{2n}\int_0^{2n}\sum_{k\in\Z}
\Bigl\{G_n^+\bigl(\theta+2n(k+1)\bigr)-G_n^+\bigl(\theta+2nk)\bigr)\Bigr\}\;\!\d\theta\\
&=\varsigma c_{-1}\Bigl\{G_n^+(\infty)-G_n^{+}(-\infty)\Bigr\}\\
&=\varsigma c_{-1}(\e^{\pi n}-\e^{-\pi n}).
\end{align*}
Finally, by collecting the result obtained so far and by using the cyclicity of the traces one gets
\begin{align*}
& \Trace_n\big(\W_{in,\kappa;\frac{1}{2},0}^- \W_{in,\kappa;\frac{1}{2},0}^{-*}
- \W_{in,\kappa;\frac{1}{2},0}^{-*} \W_{in,\kappa;\frac{1}{2},0}^-\big) \\
& = \Trace_n\Big(\Xi_{in}(D)^{-1} \Xi_{\frac{1}{2}}(D)\Big(\W_{in,\kappa;\frac{1}{2},0}^- \W_{in,\kappa;\frac{1}{2},0}^{-*}
- \W_{in,\kappa;\frac{1}{2},0}^{-*} \W_{in,\kappa;\frac{1}{2},0}^-\Big)\Xi_{\frac{1}{2}}(D)^* \Xi_{in}(D)\Big) \\
& =  - \varsigma c_{-1}(\e^{\pi n}-\e^{-\pi n}).
\end{align*}

For the computation of $c_{-1}$ it is enough to observe that
\begin{align*}
F_{in,\kappa}(x) & = \frac{-1}{(1 -\varsigma \e^{- \pi n}\e^{2inx})(1-\varsigma \e^{\pi n}\e^{2inx})} \\
& = \overline{\varsigma}\e^{-\pi n} \e^{-2inx}\;\! \frac{1}{1-\varsigma \e^{-\pi n}\e^{2inx}}\ \frac{1}{1-\overline{\varsigma} \e^{-\pi n}\e^{-2inx}} \\
& = \overline{\varsigma}\e^{-\pi n} \e^{-2inx}\;\! \sum_{j=0}^\infty \big(\overline{\varsigma} \e^{-\pi n}\e^{2inx}\big)^j\;\!
\sum_{k=0}^\infty \big(\varsigma \e^{-\pi n}\e^{-2inx}\big)^k,
\end{align*}
from which one infers by considering the diagonal sum that
$$
c_{-1} = \overline{\varsigma}\e^{-\pi n} \sum_{j=0}^\infty \big(\e^{-\pi n}\big)^{2j} = \overline{\varsigma} \frac{\e^{-\pi n}}{1-\e^{-2\pi n}}
=\overline{\varsigma}\frac{1}{\e^{\pi n}-\e^{-\pi n}}\ .
$$
The second equality of the statement follows then directly.
\end{proof}

\section{The semi-Fredholm case (II)}\label{sec_semi_II}

In this section we consider the pair of operators $(H_{in,\kappa},H_{m',\kappa'})$
with $n>0$, $\kappa\in \T$, $m'\in(0,1)$ and $\kappa'\in\R$.
For such operators the wave operator $\W^-_{in,\kappa;m',\kappa'}$ is a partial isometry with an
infinite dimensional cokernel $\I_{\rm p}(H_{in,\kappa})L^2(\R)$ and a finite
dimensional kernel $\I_{\rm p}(H_{m',\kappa'})L^2(\R)$.

In order to deal with this situation, let us first introduce the $C^*$-subalgebra $\A_{n}$ of $C_{\rm b}(\R)$ defined by
\begin{equation*}
\A_{n}:=\Big\{f\in C_{b}(\R)\mid \exists f^{\pm}\in C_{\frac{\pi}{n}}(\R)\ \text{s.t.}\ \lim_{R\to+\infty}
\sup_{\pm x>R}|f(x)-f^{\pm}(x)|= 0\Big\}.
\end{equation*}
Note that for any $f\in \A_n$ the two functions $f^{\pm}$ are unique,
and therefore we can define a map $q_n:\A_{n}\to C_{\frac{\pi}{n}}(\R)\oplus C_{\frac{\pi}{n}}(\R)$
by setting $q_n(f):=(f^-,f^+)$. Obviously, $q_n$ is a surjective *-homomorphism with $\ker q_ n=C_0(\R)$.
Hence, for any invertible $f\in\A_{n}$, the quantity
\begin{equation*}
\wn_{\A_{n}}(f):=\big(\wn_n(f^-),\wn_n(f^+)\big)\in \Z\oplus\Z
\end{equation*}
is well-defined, and if $f,g\in \A_n$ are homotopic, then $\wn_{\A_{n}}(f) = \wn_{\A_{n}}(g)$.
It is now easily observed that the function $\SS_{in,\kappa;m',\kappa'}$ introduced in \eqref{eq_S}
belongs to $\A_n$. By a direct computation one gets
\begin{equation}\label{qS}
q_n\big(\SS_{in,\kappa;m',\kappa'}\big)
= \Big(\e^{-i\pi(in-m')}\frac{(1-\varsigma \e^{-\pi n }\e^{2in\cdot})}{(1-\varsigma \e^{\pi n}\e^{2in\cdot})},
\e^{-i\pi(in+m')}\frac{(1-\varsigma \e^{-\pi n}\e^{2in\cdot})}{(1-\varsigma \e^{\pi n}\e^{2in\cdot})} \Big)
\end{equation}
and $\wn_{\A_n}\bigl(\SS_{in,\kappa;m',\kappa'}\bigr)=(-1,-1)$.
Note that the equality of these two winding numbers comes from the very simple form of the function $\SS_{in,\kappa;m',\kappa'}$,
the algebra $\A_n$ contains also elements having two asymptotic functions $f^\pm$ with different winding numbers.

Let us then define the unital $C^*$-algebra $\E_{\A_n}$ by
\begin{equation*}
\E_{\!\A_n}:=C^*\Big(a(D)b(X)\mid a\in C_0\big([-\infty,+\infty)\big),b\in \A_n\Big)^+,
\end{equation*}
and the ideal $\J_{\!\A_n}$ by
\begin{equation*}
\J_{\!\A_n}:=C^*\Big(a(D)b(X)\mid a\in C_0 (\R),b\in \A_n\Big).
\end{equation*}
Recalling that $\K_\R\equiv \K\big(L^2(\R)\big)$ we obtain the $2$-link ideal chain
$$
\K_\R\subset \J_{\!\A_n}\subset\E_{\!\A_n}
$$
and the following list of quotients:
\begin{center}
\begin{tabular}{cc}
\hline
Quotients & Images of $a(D)b(X)$ under\\
&  the corresponding quotient maps\\
\hline
$\E_{\!\A_n}/\J_{\!\A_n}\cong \A_n$ & $a(-\infty) b(\cdot)$\\
$\E_{\!\A_n}/\K_\R\cong \E_n\oplus\A_n\oplus\E_n$ & $\big(a(D)b^{-}(X), a(-\infty)b(\cdot) , a(D)b^+(X)\big)$\\
$\J_{\!\A_n}/\K_\R\cong \J_n\oplus\J_n$ & $\big(a(D)b^{-}(X), a(D)b^+(X)\big)$\\
\hline
\end{tabular}
\end{center}
where $\E_n$ and $\J_n$ were introduced in \eqref{eq_En} and \eqref{eq_Jn} respectively.

For the present set of parameters the wave operator $\W^-_{in,\kappa;m',\kappa'}$
provided in \eqref{expression in line} clearly belongs to the algebra $\E_{\!\A_n}$.
In addition, if we denote by $q_{{}_{\J_{\!\A_n}}}$ the quotient map
$q_{{}_{\J_{\!\A_n}}}:\E_{\!\A_n}\to \A_n$ then one has
$q_{{}_{\J_{\!\A_n}}}(\W^-_{in,\kappa;m',\kappa'})=\SS_{in,\kappa;m',\kappa'}$.
On the other hand, one also has
$$
[\W^{-}_{in,\kappa;m',\kappa'},\W^{-*}_{in,\kappa;m',\kappa'}] = \I_{\rm p}(H_{m',\kappa'})-\I_{\rm p}(H_{in,\kappa})
\in \J_n+ \K_\R\subset \J_{\!\A_n}.
$$
These equalities mean that if one considers the index map related to the short exact sequence
$0 \to \J_{\!\A_n}\to \E_{\!\A_n}\to \A_n\to 0$, then
$$
\ind[\SS_{in,\kappa;m',\kappa'}]_1 = \big[\I_{\rm p}(H_{m',\kappa'})\big]_0-\big[\I_{\rm p}(H_{in,\kappa})\big]_0.
$$

In order to get a meaningful numerical equality, the images of the corresponding operators through additional
quotient maps have to be considered. More precisely, let us consider
the quotient map $q_{n}:\A_{n}\to C_{\frac{\pi}{n}}(\R)\oplus C_{\frac{\pi}{n}}(\R)$ with
$q_n\big(\SS_{in,\kappa;m',\kappa'}\big)$ already computed in \eqref{qS},
and the quotient map $q_{{}_{\K_{\R}}}:\J_{\!\A_n}\to \J_n\oplus\J_n$ with
$$
q_{{}_{\K_{\R}}}\big(\I_{\rm p}(H_{m',\kappa'})-\I_{\rm p}(H_{in,\kappa})\big) = \big(-\I_{\rm p}(H_{in,\kappa}),-\I_{\rm p}(H_{in,\kappa})\big).
$$
If we still define the map $\Trace_{\A_n}:\J_{\!\A_n}\to \C^2$ by
$\Trace_{\A_n}:=\big(\Trace_n\oplus\Trace_n\big)\circ q_{{}_{\K_\R}}$
one directly deduces from Theorem \ref{index theorem 4} the following statement:

\begin{theorem}\label{index theorem 6}
Let $n>0$, $\kappa\in \T$, $m'\in(0,1)$ and $\kappa'\in\R$. Then one has
\begin{equation*}
\wn_{\A_n}\big(\SS_{in,\kappa;m',\kappa'}\big)=(-1,-1)=\Trace_{\A_n}\big([\W^{-}_{in,\kappa;m',\kappa'},\W^{-*}_{in,\kappa;m',\kappa'}]\big).
\end{equation*}
\end{theorem}

Unfortunately, the information about the kernel of $\W^{-}_{in,\kappa;m',\kappa'}$ is lost in the previous construction.
In order to recover this information, it seems necessary
to use a comparison operator, which is simpler and better known.
The aim of this construction is to hide the infinite dimensional cokernel.
In the present situation, the best choice seems to be $\W^{-}_{in,\kappa;\frac{1}{2},0}$ since this operator has already been studied
in the previous section and since it has a trivial kernel and a cokernel given by $\I_{\rm p}(H_{in,\kappa})L^2(\R)$.
We shall then call the \emph{relative index theorem of $\W^{-}_{in,\kappa;m',\kappa'}$ with respect to $\W^{-}_{in,\kappa;\frac{1}{2},0}$}
the index theorem for the partial isometry $\W^{-*}_{in,\kappa;\frac{1}{2},0} \W^{-}_{in,\kappa;m',\kappa'}$.
By using standard properties of the wave operators and the chain rule one infers that
$$
\W^{-*}_{in,\kappa;\frac{1}{2},0} \W^{-}_{in,\kappa;m',\kappa'}
= \W^{-}_{\frac{1}{2},0;in,\kappa} \W^{-}_{in,\kappa;m',\kappa'}
= \W^{-}_{\frac{1}{2},0;m',\kappa'}.
$$
The latter operator is a partial isometry with kernel $\I_{\rm p}(H_{m',\kappa'})L^2(\R)$ and trivial
cokernel. It is thus a Fredholm operator and its analysis can be performed in the setting of Section \ref{sec_Fredholm}.
We then get:

\begin{theorem}\label{thm_plugged}
For any $n>0$, $\kappa\in \T$, $m'\in(0,1)$ and $\kappa'\in\R$
one has
\begin{equation*}
q_{{}_{\K_\R}}\big( \W^{-*}_{in,\kappa;\frac{1}{2},0} \W^{-}_{in,\kappa;m',\kappa'}\big)
=\Gamma_{\frac{1}{2},0;m',\kappa'}^{\triangle}
\ \in C(\triangle)
\end{equation*}
with
$$
\Gamma_{\frac{1}{2},0;m',\kappa'}^{\triangle}
:= \Big(\e^{i\frac{\pi}{2}(m'-\frac{1}{2})}\Xi_{m'}(-\cdot)\Xi_\frac{1}{2}(\cdot), \SS_{\frac{1}{2},0;m'.\kappa'},
\e^{-i\frac{\pi}{2}(m'+\frac{1}{2})}\Xi_{-m'}(-\cdot)\Xi_{\frac{1}{2}}(\cdot) \Big).
$$
In addition,
\begin{equation*}
\wn_{\triangle}\big(\Gamma_{\frac{1}{2},0;m',\kappa'}^{\triangle}\big)=
-\text{number of eigenvalues of}\ H_{m',\kappa'}=
-\Index\big(\W_{\frac{1}{2},0;m',\kappa'}^{-}\big).
\end{equation*}
\end{theorem}

\section{The almost periodic case}\label{sec_ap}

In this section we consider the pairs of parameters $(in,\kappa)$ and $(in',\kappa')$ with $n,n'\in (0,\infty)$ and $\kappa,\kappa'\in \T$.
Even though the kernel and the cokernel of the wave operator $\W^-_{in,\kappa;in',\kappa'}$ are infinite-dimensional in this case,
we shall exhibit an index theorem by using a $C^*$-algebras generated by \emph{almost periodic pseudo-differential operators}.

Let us start by recalling some basic facts about almost periodic functions and pseudo-differential operators
with almost periodic coefficients based on \cite[Sec.~1]{Shubin1}.
A continuous function $f:\R\to\C$ is said to be \emph{uniformly almost periodic} if any sequence of translations
$\{f(\cdot+\alpha_k)\}_{k=1}^{\infty}$ contains a subsequence that is a uniformly convergent on $\R$.
Then, the set $\CAP(\R)$ of all uniformly almost periodic functions coincides with the $C^*$-subalgebra of $C_{\rm b}(\R)$
generated by exponentials $\{\e^{i\lambda\cdot}\}_{\lambda\in\R}$.

For any $f\in\CAP(\R)$, the limit
\begin{equation*}
M(f):=\lim_{T\to+\infty}\frac{1}{2T}\int_{-T}^{T}f(x)\;\!\d x
\end{equation*}
exists and is called \emph{the mean value} of $f$.
Note that $M(f)=\frac{1}{L}\int_{0}^{L}f(x)\;\!\d x$ holds for any  continuous periodic function $f$ with a period $L>0$.
Moreover, the completion of $\CAP(\R)$ with respect to the norm associated with the inner product defined by $(f,g)\mapsto M(\overline{f}g)$
is called \emph{the space of Besicovitch a.p.~functions} and is denoted by $B^2(\R)$.
Since $\{\e^{i\lambda\cdot}\}_{\lambda\in\R}$ is an orthonormal basis for $B^2(\R)$,
$B^2(\R)$ is a non-separable Hilbert space.

The dual group $\R_{\Bohr}$ of $\R_{\discrete}$ (the group $\R$ with the discrete topology)
is called \emph{the Bohr compactification of}  $\R$.
$\R_{\Bohr}$ is a compact commutative group containing $\R$ as a dense subgroup. There exists then an isomorphism
\begin{equation*}
\CAP(\R)\cong C(\R_{\Bohr}).
\end{equation*}
In addition, if $\mu_{\Bohr}$ denotes the normalized Haar measure on $\R_{\Bohr}$, then one has for any $f\in\CAP(\R)$,
\begin{equation*}
M(f)=\int_{\R_{\Bohr}}f \;\!\d\mu_{\Bohr},
\end{equation*}
where the function $f$ in the right-hand side stands for the extension of $f$ by continuity on $\R_{\Bohr}$.
There is then a natural isomorphism
$$
L^2(\R_{\Bohr})\equiv L^2(\R_{\Bohr},\d\mu_{\Bohr})\cong B^2(\R).
$$

We now recall the topological invariant of invertible almost periodic functions inspired by \cite[Sec.~2]{CDSS},
which is a natural generalization of the winding number for invertible elements of $C(\S)$.
Let $f\in \CAP(\R)$ be invertible and write it as $f(x)=|f(x)|e^{i\sigma(x)}$ for some continuous function $\sigma:\R\to\R$.
Then the limit
\begin{equation*}
\wn_{\ap}\bigl(f\bigr):=\lim_{T\to+\infty}\frac{\sigma(T)-\sigma(-T)}{2T}
\end{equation*}
exists, and defines a homotopy invariant on the set of all invertible elements in $\CAP(\R)$.
Similar to the usual winding numbers, $\wn_{\ap}$ defines a group homomorphism
from the set of all invertible elements in $\CAP(\R)$ to $\R$.
Note that if $f$ is $L$-periodic for some $L>0$, then one obtains $\wn_{\ap}\bigl(f\bigr)=2\pi L^{-1}\wn_{L}\bigl(f\bigr)$.

Let us now construct a $C^*$-algebraic framework for treating the wave operator $\W^-_{in,\kappa;in',\kappa'}$.
Let $\E_{\ap}$ be the $C^*$-subalgebra defined by
$$
\E_{\ap}:=\Big(a(D)b(X)\mid a\in C_0\big([-\infty,+\infty)\big), b\in\CAP(\R) \Big)^+.
$$
As for the other cases, $\E_{\ap}$ is an extension of $\CAP(\R)$ by the ideal
$$
\J_{\ap}:= \Big(a(D)b(X)\mid a\in C_{0}(\R), b\in \CAP(\R) \Big).
$$
By looking at the explicit expression for $W^{-}_{in,\kappa;in',\kappa'}$ provided in \eqref{expression in line}
one directly infers that this operator belongs to $\E_{\ap}$ for any $n,n'\in (0,\infty)$ and $\kappa,\kappa'\in \T$.

To obtain a trace on the ideal $\J_{\ap}$ we recall the construction of a specific II${}_{\infty}$-factor $\A_{\Bohr}$
and its relations to pseudo-differential operators with almost periodic coefficients, as introduced in \cite[Sec.~3.3]{Shubin1}.
Let us consider the Hilbert space
$$
\H_{\Bohr}:=B^2(\R)\otimes L^2(\R)\cong L^2(\R_{\Bohr}\times\R),
$$
and for any $\xi \in \R$ let $T_\xi$ is the translation given by $T_{\xi}f:=f(\cdot-\xi)$ for $f\in L^2(\R)$ or $f\in B^2(\R)$.
The factor $\A_{\Bohr}$ is defined to be the von Neumann algebra on $\H_{\Bohr}$ generated by two families of operators
\begin{equation*}
\{e^{i\xi X}\otimes e^{i\xi Y}\mid \xi\in\R\}\ \text{and}\ \{\one\otimes T_\xi \mid \xi\in\R\}.
\end{equation*}
Then, it is known that the commutant $\A_{\Bohr}'$ is generated by
\begin{equation*}
\{T_{\xi}\otimes T_{-\xi}\mid \xi\in\R\}\ \text{and}\ \{e^{i\xi X}\otimes \one \mid \xi\in\R\},
\end{equation*}
and that $\A_{\Bohr}$ is indeed a II$_{\infty}$-factor.

For $m\leq 0$ let us now denote by $\APS^m$ the set of usual symbols $(x,\xi)\mapsto a(x,\xi)$ of degree $m$ which are uniformly almost periodic in the first variable,
namely $\partial_\xi^\alpha \partial_x^\beta a(\cdot,\xi)\in\CAP(\R)$ for all $\xi\in \R$, $\alpha,\beta \in \N$.
The set of corresponding pseudo-differential operators $a(X,D)$ is denoted by $\APL^m$.
Then, we can define a map ${}^{\#}: \APL^0\to \A_{\Bohr}$ given by
\begin{equation*}
a(X,D)^{\#}:=a(X+Y,D_y),
\end{equation*}
where in the r.h.s.~the operator $X$ refer to the multiplication by the variable in $\R_{\Bohr}$ while
$Y$ and $D_y$ refer to the variable on $\R$.
Note that the map ${}^{\#}$ is an injective *-homomorphism from $\APL^0$ to $\A_{\Bohr}$.
Then, for any symbol $a$ of degree $m<-1$ an explicit trace formula is proposed in \cite[Prop.~3.3]{Shubin1}~:
\begin{equation}\label{eq_Trace}
\Trace_{\Bohr}\big(a(X,D)^\#\big):=(2\pi)^{-1}\int_{\R_{\Bohr}\times \R} a(x,\xi)\;\!\d \mu_{\Bohr}(x) \;\!\d \xi.
\end{equation}

Motivated by the above expression, let us now define a suitable trace $\Trace_{\ap}$ on $\J_{\ap}$.
For that purpose, let us consider $a\in C_0(\R)$ with $|a(\xi)|\leq C\langle\xi\rangle^{-1-\varepsilon}$ for some $C,\varepsilon>0$ and all $\xi \in \R$,
and let $b\in \CAP(\R)$. Clearly $a(D)b(X)$ belongs to $\J_{\ap}$. We then set
\begin{equation*}
\Trace_{\ap}\big(a(D)b(X)\big)=\int_{\R}a(\xi)\;\!\d\xi \times M(b).
\end{equation*}
Note that the normalization used in this expression is not the same as in \eqref{eq_Trace}

Before stating our last index theorem, observe that if $a(D)b(X)\in\J_n\subset\J_{\ap}$, then
the trace $\Trace_{\ap}$ and the trace $\Trace_n$ are connected.
Indeed, for $a\in C_0(\R)$ with $|a(\xi)|\leq C\langle\xi\rangle^{-1-\varepsilon}$ for some $C,\varepsilon>0$ and all $\xi \in \R$,
and for $b\in C_{\frac{\pi}{n}}(\R)$ one has
\begin{align*}
\Trace_{\ap}\big(a(D)b(X)\big) & =\int_{\R}a(\xi)\;\!\d\xi \times M(b) \\
& = \int_{\R}a(\xi)\;\!\d\xi \times \frac{n}{\pi}\int_0^\frac{\pi}{n} b(x)\;\!\d x \\
& = 2n \Trace_n \big(a(D)b(X)\big)
\end{align*}
with $\Trace_n \big(a(D)b(X)\big)$ computed in \eqref{formula of tracen}.
With the help of these equalities one can now easily get:

\begin{theorem}\label{index theorem 5}
Let $n,n'\in (0,\infty)$ and $\kappa,\kappa'\in \T$. Then one has
\begin{equation}\label{eq_ap}
\wn_{\ap}\big(\SS_{in,\kappa;in',\kappa'}\big)=-2(n-n')=\Trace_{\ap}\big([\W^{-}_{in,\kappa;in',\kappa'},\W^{-\ast}_{in,\kappa;in',\kappa'}]\big).
\end{equation}
\end{theorem}

\begin{proof}
The expression for the function $\SS_{in,\kappa;in',\kappa'}$ has been provided in \eqref{eq_S}
and for any $x\in \R$ one has
\begin{align*}
\SS_{in,\kappa;in',\kappa'}(x)
& = \e^{\pi(n-n')}\frac{1-\varsigma \e^{-\pi n}\e^{2inx}}{1-\varsigma \e^{\pi n}\e^{2inx}}
\ \frac{1-\varsigma' \e^{\pi n'}\e^{2in'x}}{1-\varsigma' \e^{-\pi n'}\e^{2in'x}} \\
& = \overline{\varsigma}\varsigma' \e^{-2i(n-n')x}\frac{1-\varsigma \e^{-\pi n}\e^{2inx}}{1-\overline{\varsigma} \e^{-\pi n}\e^{-2inx}}
\ \frac{1-\overline{\varsigma}' \e^{-\pi n'}\e^{-2in'x}}{1-\varsigma' \e^{-\pi n'}\e^{2in'x}}.
\end{align*}
Clearly this function belongs to $\CAP(\R)$ and one infers from the above expression the first equality in \eqref{eq_ap}.

Recall now from Sections \ref{sec_model} and \ref{sec_semi_I} that
$$
[\W^{-}_{in,\kappa;in',\kappa'},\W^{-\ast}_{in,\kappa;in',\kappa'}] = \I_{\rm p}(H_{in',\kappa'})-\I_{\rm p}(H_{in,\kappa})
$$
with $\I_{\rm p}(H_{in',\kappa'})\in \J_{n'}$ and $\I_{\rm p}(H_{in,\kappa})\in \J_n$.
In addition, since $\Trace_{n'} \big(\I_{\rm p}(H_{in',\kappa'})\big)=1$ and
$\Trace_{n} \big(\I_{\rm p}(H_{in,\kappa})\big)=1$, we infer from the relation between the traces $\Trace_{n'}$
(or $\Trace_{n}$) and $\Trace_{\ap}$ mentioned before the statement that
$$
\Trace_{\ap} \big(\I_{\rm p}(H_{in',\kappa'})\big) = 2n'
\quad \hbox{ and }\quad
\Trace_{\ap} \big(\I_{\rm p}(H_{in,\kappa})\big)=2n.
$$
The second equality of the statement follows directly from this computation.
\end{proof}

\begin{remark}\label{rem_dens}
If we just look at the expression for the wave operators provided in \eqref{expression in line}
nothing special can be said about their kernel or cokernel. However, if we remember that
these subspaces correspond to the subspaces generated by the eigenfunctions of $H_{in',\kappa'}$ and $H_{in,\kappa}$ respectively,
then one more relation can be established.
For that purpose, let us set $\lambda_{in,\kappa,0}:=-4\exp(-\arg(\varsigma)/n)$ for $n>0$ and $\kappa\in \T$,
which corresponds to one of the eigenvalues of $H_{in,\kappa}$.
As shown in Theorem \ref{spectral theory}.(ii), the set of all eigenvalues of $H_{in,\kappa}$ is then given by
$\{\lambda_{in,\kappa,0}\e^{\frac{2\pi}{n} j}\mid j\in \Z\}$.
Now, let us define the function $N_{in,\kappa}:\R_+\to\N$ for $T>0$ by
\begin{equation*}
N_{in,\kappa}(T):=\#\Big\{\text{eigenvalues of}\ H_{in,\kappa}\ \text{in the interval}\ [-4\e^{2\pi T},-4\e^{-2\pi T}]\Big\}
\end{equation*}
We then infer that
\begin{equation*}
N_{in,\kappa}(T)=2nT+O(1)
\end{equation*}
which directly implies that $\frac{N_{in,\kappa}(T)}{N_{in',\kappa'}(T)}\to\frac{n}{n'}$ when $T\to+\infty$.
Thus, the relative density of eigenvalues for the pair of operator $(H_{in,\kappa},H_{in',\kappa'})$ also
corresponds to the number $\frac{n}{n'}$, in the sense introduced above.
\end{remark}

\end{document}